\documentclass[12pt,fleqn]{article}

\usepackage{amsmath,amsthm,enumerate,amssymb}
\usepackage[utf8]{inputenc}

\usepackage[english]{babel}
\usepackage{graphicx}
\usepackage{amsthm}

\newtheorem{theorem}{Theorem}[section]

\newtheorem{lemma}[theorem]{Lemma}
\newtheorem{corollary}[theorem]{Corollary}

\newtheorem{remark}[theorem]{Remark}

\newtheorem*{acknowledgement}{Acknowledgement}

\newcommand{\tr}{{\rm Tr\hskip -0.2em}~}

\begin{document}

\title{Trace functions\\with\\applications in quantum physics}
\author{Frank Hansen}
\date{July 2, 2013}

\maketitle

\begin{abstract}

We consider both known and not previously studied trace functions with applications in quantum physics. By using perspectives we obtain convexity statements for different notions of residual entropy, including the entropy gain of a quantum channel as studied by Holevo and others.

We give new and simplified proofs of the Carlen-Lieb theorems concerning concavity or convexity of certain trace functions by making use of the theory of operator monotone functions. We then apply these methods in a study of new types of trace functions.
\\[1ex]
\end{abstract}

\section{Introduction and first results}

Consider a quantum system in which an observable $ A $ can be written as a sum $ A=A_1+\cdots+ A_k $ of a number of components $ A_1,\dots,A_k. $ If the components correspond to isolated subsystems then the total quantum entropy of the system $ S(A)=-\tr A\log A $ is equal to the sum of the entropies of each subsystem. In the general case we may define the residual entropy
\[
\varphi(A_1,\dots,A_k)=S(A)-\sum_{i=1}^k S(A_i)\qquad A=A_1+\cdots+A_k 
\]
as the difference between the total entropy of the system and the sum of the entropies of each subsystem; although it is a negative quantity. 

Another type of residual entropy is the entropy gain over a quantum channel studied by Holevo and others \cite{kn:holevo:2011:1, kn:holevo:2012:1},
\[
A\to S\bigl(\Phi(A)\bigr)-S(A),
\]
where $ \Phi $ is a quantum channel represented by a completely positive trace preserving linear map.

\begin{theorem} \label{main theorem}
Consider $ n\times n $ matrices $ A $ and $ n\times m $ matrices $ K. $ The trace function
\[
\varphi(A)=-\tr K^*AK\log(K^*AK)+\tr K^*(A\log A)K
\]
is convex in positive definite $ A $ for arbitrary $ K. $ 
\end{theorem}

\begin{proof}
The function $ f(t)=t\log t $ defined for $ t>0 $ is operator convex. It is well-known but may be derived from~\cite[Theorem 2.4]{kn:hansen:1982} since $ f(0)=0, $ and $ \log t $ is operator monotone. The perspective function,
\[
g(t,s)=s f(ts^{-1})=t\log t-t\log s\qquad t,s>0,
\]
is therefore operator convex as a function of two variables ~\cite[Theorem 2.2]{kn:effros:2009:1}. Consider the Hilbert space $ \mathcal H=M_{n\times m} $ equipped with inner product given by $ (X,Y)=\tr Y^*X $ for matrices $ X, Y\in M_{n\times m}  $ and let $ L_A $ and $ R_B $ denote left and right multiplication with $ A\in M_n $ and $ B\in M_m $ respectively. If $ A $ and $ B $ are positive definite  matrices then $ L_A  $ and $ R_B $ are positive definite commuting operators on $ \mathcal H. $ Operator convexity of the perspective function $ g(t,s) $ is equivalent to convexity of the map
\[
\begin{array}{rl}
(A,B)&\to\tr K^*\bigl(L_{A\log A}-L_A R_{\log B}\bigr)(K)\\[1.5ex]
&=\tr\bigl(K^* (A\log A) K - K^* A K\log B\bigr) \qquad A,B>0
\end{array}
\]
for every $ K\in M_{n\times m} $ cf.~\cite[Theorem 1.1]{kn:hansen:2006:3}. 
The statement of the theorem now follows by replacing $ B $ with $ K^*AK $ in the above expression.
\end{proof}

\begin{corollary}
The residual entropy 
\[
\varphi(A_1,\dots,A_k)=-\tr A\log A+\sum_{i=1}^k \tr A_i\log A_i\qquad A=A_1+\cdots+A_k 
\]
is a convex function in  positive definite $ n\times n $ matrices $ A_1,\dots,A_k. $
\end{corollary}

\begin{proof}
We apply Theorem~\ref{main theorem} to block matrices of the form
\[
A=\begin{pmatrix}
A_1 &  0  & \cdots &  0\\
0       & A_2 & & 0\\
\vdots & & \ddots\\
0       &  0    &      & A_k
\end{pmatrix}       
\qquad\text{and}\qquad
K=\begin{pmatrix}
I & 0  & \cdots & 0\\
I & 0  & \cdots & 0\\
\vdots & \vdots & & \vdots\\
I & 0    & \cdots & 0            
\end{pmatrix},
\]
and since the entry in the first row and the first column of the block matrix
\[
-K^*AK\log(K^*AK)+K^*(A\log A)K
\]
is calculated to
\[
-(A_1+\cdots+A_k)\log\bigl(A_1+\cdots+A_k\bigr)+\sum_{i=1}^k A_i\log A_i
\]
the statement of the corollary follows.
\end{proof}

It is actually much easier to obtain the above result by expressing the residual entropy as a sum of relative entropies. We may however obtain other results by carefully choosing the arbitrary matrix $ K $ in Theorem~\ref{main theorem}.

\begin{corollary}
Consider the entropy gain 
\[
\varphi(A)= S\bigl(\Phi(A)\bigr)-S(A)
\]
over a quantum channel $ \Phi, $ where the channel is represented by a completely positive trace preserving linear map $ \Phi. $ The entropy gain $ \varphi(A) $ is a convex function in $ A. $
\end{corollary}

\begin{proof} A completely positive trace preserving linear map $ \Phi\colon M_n\to M_m $ is of the form
\[
\Phi(A)=\sum_{i=1}^k a_i^* A a_i
\]
where the so-called Kraus matrices $ a_1,\dots,a_k\in M_{n\times m} $ satisfy
\[
a_1 a_1^*+\cdots+a_k a_k^*=1.
\]
We now apply Theorem~\ref{main theorem} by setting
\[
A=\begin{pmatrix}
A &  0  & \cdots &  0\\
0   & A & & 0\\
\vdots & & \ddots\\
0   &  0    &     & A
\end{pmatrix}       
\qquad\text{and}\qquad
K=\begin{pmatrix}
a_1 & 0  & \cdots & 0\\
a_2 & 0  & \cdots & 0\\
\vdots & \vdots & & \vdots\\
a_k & 0    & \cdots & 0            
\end{pmatrix}.
\]
The entry in the first row and the first column of the block matrix
\[
-K^*AK\log(K^*AK)+K^*(A\log A)K
\]
is calculated to $ -\Phi(A)\log\Phi(A)+\Phi( A\log A). $ Since $ \Phi $ is trace preserving it follows that the entropic map
\[
A\to S\bigl(\Phi(A)\bigr)+\tr\Phi(A\log A)=S\bigl(\Phi(A)\bigr)-S(A)
 \]
 is convex.
\end{proof}

\begin{corollary}
The entropy gain
\[
\varphi(A_1,\dots,A_k)=S\bigl(\Phi_1(A_1)+\cdots+\Phi_k(A_k)\bigr)-\sum_{i=1}^k S(A_i)
\]
of $ k $ positive definite quantities observed through $ k $ quantum channels $ \Phi_1,\dots,\Phi_k $ is a  convex function
in $ A_1,\dots,A_k. $ 
\end{corollary}

\begin{proof}
The statement is obtained as in the above corollary by considering suitable block matrices, where each block corresponds to a single quantum channel. We leave the details to the reader. 
\end{proof}

\section{Carlen-Lieb trace functions}

We give new proofs of some of the statements in \cite{kn:carlen:2008} without using variational methods.

\begin{theorem} (Carlen-Lieb)
The trace function
\[
(A,B)\to\tr (A^p+B^p)^{1/r}\qquad 0<p\le r\le 1
\]
is concave for positive definite matrices $ A $ and $ B. $
\end{theorem}

\begin{proof}
The function
\[
f(t)=(t^p+1)^{1/p}\qquad t>0
\]
is operator monotone.
 Indeed, if $ z=r e^{i\theta} $ with $ 0<\theta<\pi $ then $ z^p=r^p e^{i p\theta}. $ Since we add a positive constant it is plain that the argument of $ z^p +1 $ is less than $ p\theta $ but still positive. The argument of $ f(z) $ is therefore between zero and $ p\theta\le\theta <\pi. $  We have shown that the analytic continuation of $ f $ to the complex upper half plane has positive imaginary part, thus $ f $ is operator monotone. 

The perspective function
\[
(t,s)\to s f(ts^{-1})=s \bigl(t^ps^{-p}+1)^{1/p}=(t^p+s^p)^{1/p}
\]
is therefore operator concave, cf.~\cite[Theorem 2.2]{kn:effros:2009:1} and so is the function,
\[
g(t,s)=(t^p+s^p)^{1/r}\qquad t,s>0,
\]
that appears by composing with the operator monotone and operator concave function $ t\to t^{p/r}. $

The left and right multiplication operators $ L_A  $ and $ R_B $ are positive definite commuting operators on the Hilbert space $ \mathcal H=M_n $ equipped with the inner product $ (A,B)=\tr B^* A. $ It follows that the (super) operator mapping
\[
(A,B)\to \big(L_A^p+R_B^p\bigr)^{1/r}
\]
is concave according to the preceding remark. The trace function
\begin{equation}\label{more general trace function}
(A,B)\to\tr K^* \big(L_A^p+R_B^p\bigr)^{1/r}(K)
\end{equation}
is therefore concave by~\cite[Theorem 1.1]{kn:hansen:2006:3}. The statement now follows by choosing $ K $ as the identity matrix. Indeed, under the trace we have
\[
\tr (L_A+L_B)(A+B)^n=\tr (A+B)^{n+1}
\]
for each $ n, $ and we thus obtain
\[
\tr \big(L_A^p+R_B^p\bigr)^{1/r}(I)=\tr (A^p+B^p)^{1/r}
\]
by simple algebraic calculations.
\end{proof}

Notice that the statement in (\ref{more general trace function}) is stronger than what is obtained in the
reference~\cite{kn:carlen:2008}.

\begin{theorem}
The function
\[
f(t)=(t^p+1)^{1/p}\qquad t>0
\]
is operator convex for $ 1\le p\le 2. $
\end{theorem}

\begin{proof}
We have previously shown that $ f $  is operator monotone for $ 0<p\le 1. $ Let us calculate the representing measure. \\[1ex]

\hbox{
\begin{minipage}[c]{5em}
\centering
$f(t)$
\end{minipage}

\begin{minipage}[c]{20em}
\centering
\includegraphics[scale=0.6]{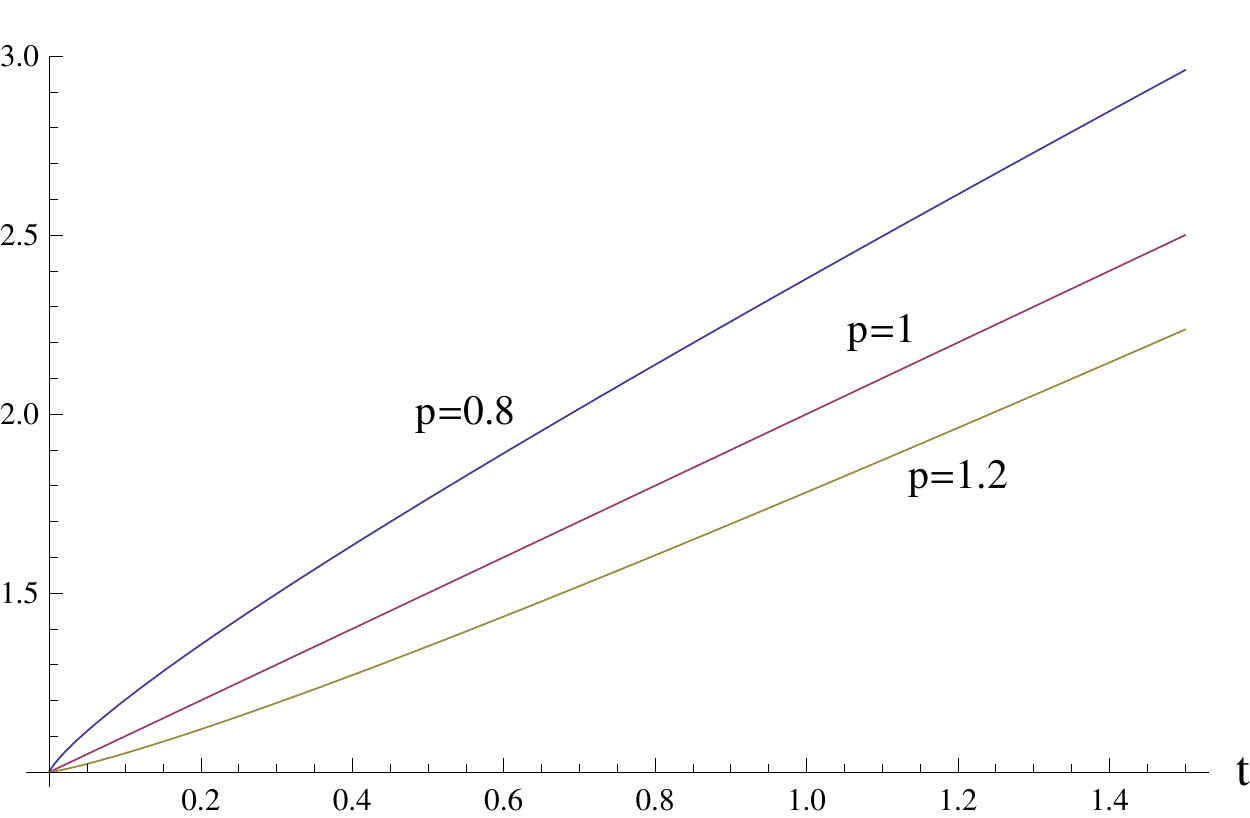}

\end{minipage}
}
{\-}\\[1ex]

We set $ z=r e^{i\theta} $ for $ r>0 $ and $ 0<\theta<\pi $ and calculate the analytic continuation of $ f, $
\[
f(r e^{i\theta})=\bigl(r^p e^{ip\theta}+1\bigr)^{1/p},
\] 
into the complex upper half plane. Let $ \arg z $ with $ 0\le\arg z<2\pi $ denote the angle between the positive $ x $-axis and the complex number $ z=x+iy. $ With this non-standard convention $ \arg z $ is an analytic function in $ \mathbf C\backslash [0,\infty), $ and the angle  $ A_p(r,\theta) $  between the positive $ x $-axis and  $ (r^p e^{ip\theta}+1)^{1/p} $  is given by 
\[
 A_p(r,\theta)=\frac{1}{p}\arg(r^p\cos p\theta +1 +i r^p \sin p\theta),
\]
and it satisfies
\[
0< A_p(r,\theta)<\theta<\pi\qquad\text{for $ 0<p\le 1, $ $ r>0, $ $ 0<\theta<\pi. $}
\]
The imaginary part of the analytic continuation of $ f $ is therefore given by
\[
\Im f(r e^{i\theta})=
(1+r^{2p}+2r^p\cos p\theta)^{1/(2p)} \sin A_p(r,\theta),
\]
and the representing measure of $ f $ is obtained as the limit
\[
\frac{1}{\pi} \lim_{\theta\to\pi} \Im f(r e^{i\theta})=
\frac{1}{\pi}(1+r^{2p}+2r^p\cos p\pi)^{1/(2p)} \sin A_p(r,\pi).
\]
It follows that
\begin{equation}\label{integral representation of p-function}
(t^p+1)^{1/p}=\beta+t+\int_0^\infty\left(\frac{\lambda}{1+\lambda^2}-\frac{1}{t+\lambda}\right) h_p(\lambda)\,d\lambda,
\end{equation}
where $ \beta $ is a constant determined by setting $ t=0 $ in equation (\ref{integral representation of p-function}), and the non-negative function $ h_p $ is given by
\[
h_p(\lambda)=\frac{1}{\pi}(1+\lambda^{2p}+2\lambda^p\cos p\pi)^{1/(2p)} \sin A_p(\lambda,\pi)\qquad \lambda> 0,
\]
cf.~\cite{kn:hansen:2013:1} for the details.
The key in the proof is the realisation that
 \[
 \pi<A_p(\lambda,\pi)<2\pi\qquad\text{for $ 1<p< 2 $ and  $ r >0, $}
 \] 
and this is so because $ \arg z<\arg (z+1)<2\pi $ when $ z $ is in the lower complex plane. It follows that
both sides in equation (\ref{integral representation of p-function}) are real analytic functions in $ p $ in the whole interval $ (0,2). $ \\[1ex] 
\hbox{
\begin{minipage}{5em}
\centering
$ h_p(\lambda) $
\end{minipage}

\begin{minipage}{20em}
\centering
\includegraphics[scale=0.5]{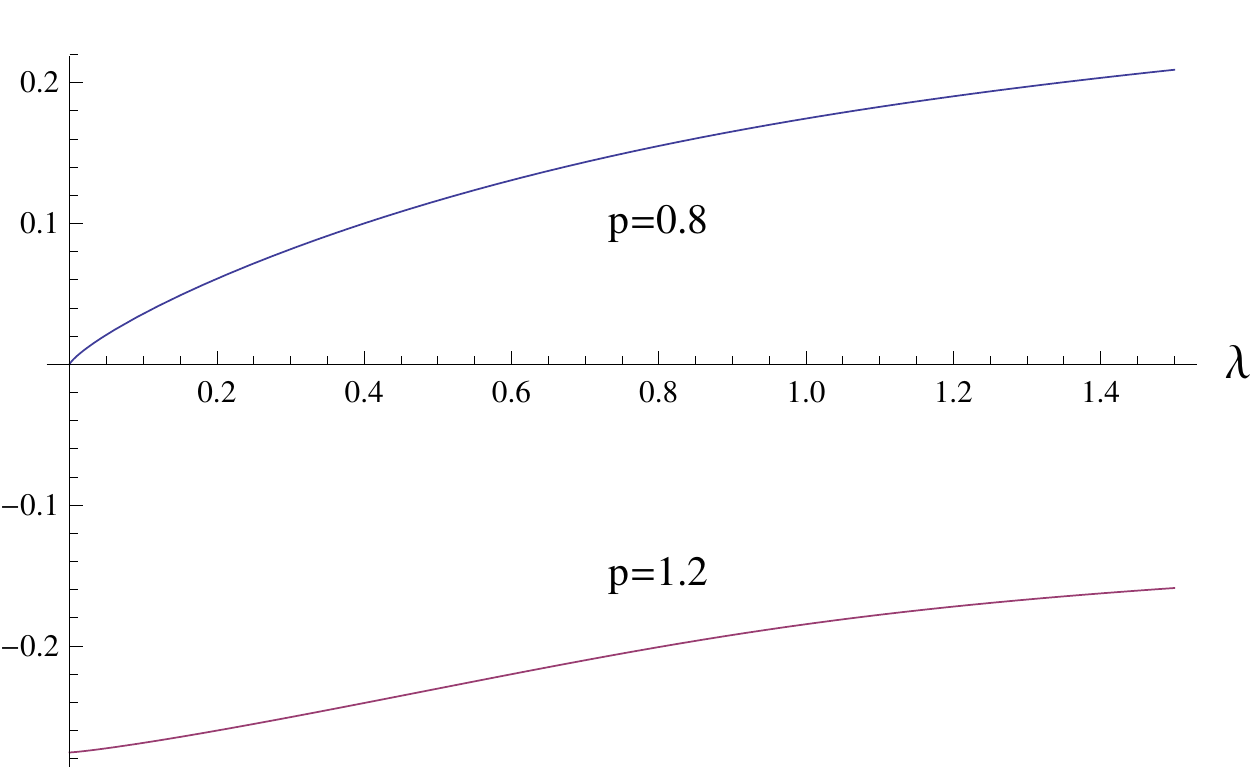}
\end{minipage}
}
{\-}\\[2ex]
The formula in (\ref{integral representation of p-function}) is consequently valid also for $ 1\le p\le 2. $
However,  for $ 1< p< 2 $ the weight function $ h_p $ is negative implying that $ f $ is operator convex. Notice that $ h_p=0 $ for $ p=1. $\end{proof}

The same line of arguments as for $ 0<p\le 1 $ applies, so we obtain:

\begin{corollary}
The trace function
\[
(A,B)\to\tr (A^p+B^p)^{1/p}\qquad 1\le p\le 2
\]
is convex for positive definite matrices $ A $ and $ B. $
\end{corollary}

\subsection{Variational inequalities}

\begin{remark}
Let $ x $ and $ y $ be positive numbers and take $ 0<p< 1. $  It is easy to prove that
\[
(x^p+y^p)^{1/p}\le \lambda^{(p-1)/p} x + (1-\lambda)^{(p-1)/p}y\qquad\text{for}\quad 0<\lambda<1
\]
with equality for $ \lambda=x^p (x^p+y^p)^{-1}. $
\end{remark}

\begin{theorem}
Let  $ 0< p<1 $  and take positive definite $ n\times n $ matrices  $ A,B. $ Then
\[
\tr (A^p+B^p)^{1/p}\le \tr\bigl( X^{(p-1)/p} A + (1-X)^{(p-1)/p} B\bigr)
\]
for each $ n\times n $ matrix $ X $ with $ 0<X<1. $ If $ A $ and $ B $ commute then there is equality for $  X=A^p (A^p+B^p)^{-1}. $
\end{theorem}

\begin{proof}
We know that the trace function $ \varphi(X,Y)=\tr (X^p+Y^p)^{1/p} $ is concave in positive definite $ X $ and $ Y. $ It is also positively homogeneous since
\[
\varphi(tX,tY)=t\varphi(X,Y)\qquad t>0.
\]
It follows that the Fréchet differential
\[
d\varphi(X,Y)(A,B)\ge\varphi(A,B)
\]
for positive definite $ X,Y,A,B, $ cf. for example \cite[Lemma 5]{kn:lieb:1973:1}.  We notice that
\[
d\varphi(X,Y)(A,B)=d_1\varphi(X,Y)A+d_2\varphi(X,Y)B
\]
by the chain rule for Fréchet differentials. By
setting $ f(t)=t^{1/p} $ and $ g(t)=t^p $ we obtain
\[
\begin{array}{rl}
d_1\varphi(X,Y)A&=\tr df(X^p+Y^p) dg(X)A=\tr f'(X^p+Y^p) dg(X)A\\[2ex]
&=\displaystyle\frac{1}{p}\tr (X^p+Y^p)^{(1-p)/p} dg(X)A
\end{array}
\]
and similarly
\[
d_2\varphi(X,Y)B=\frac{1}{p}\tr (X^p+Y^p)^{(1-p)/p} dg(Y)B.
\]
We thus derive that
\[
\tr (A^p+B^p)^{1/p}\le\frac{1}{p} \tr (X^p+Y^p)^{(1-p)/p}\bigl(dg(X)A+ dg(Y)B\bigr).
\]
Let now $ 0<X<1 $ and set $ Y=(1-X^p)^{1/p}. $ Then $ X^p+Y^p=1 $ and thus
\[
\begin{array}{rl}
\tr (A^p+B^p)^{1/p}&\le\displaystyle\frac{1}{p} \tr \bigl(dg(X)A+ dg(Y)B\bigr)\\[2.5ex]
&=\displaystyle\frac{1}{p} \tr \bigl(g'(X)A+ g'(Y)B\bigr)\\[2.5ex]
&=\displaystyle \tr\bigl( X^{p-1} A+(1-X^p)^{(p-1)/p}B\bigr).
\end{array}
\]
We may replace $ X $ with $ X^{1/p} $ since any $ 0<X<1 $ can be obtained in this way, and we obtain
\[
\tr (A^p+B^p)^{1/p}\le\tr\bigl( X^{(p-1)/p} A+ (1-X)^{(p-1)/p} B\bigr)
\]
which is the statement of the theorem.
\end{proof}

\section{New types of trace functions}

\begin{theorem}\label{main operator concave function}
Let $ 0< p\le 1.  $ The function of two variables,
\[
g(t,s)=\left\{\begin{array}{ll} 
\displaystyle\frac{t-s}{t^p-s^p}\qquad &t\ne s\\[3ex]
\frac{1}{p} t^{1-p}           &t=s,
\end{array}\right.
\]
defined for $ t,s>0, $ is operator concave.
\end{theorem}

\begin{proof}
We first prove that for $ 0\le\lambda\le 1 $ the function
\[
f_\lambda(t)=\bigl(\lambda t^p + 1-\lambda\bigr)^{1/p}\qquad t>0
\]
is operator monotone. Indeed, if $ z=r e^{i\theta} $ with $ 0<\theta<\pi, $ then $ z^p=r^p e^{i p\theta}. $ Since we add a positive constant it is plain that the argument of $ \lambda z^p +1-\lambda $ is less that $ p\theta $ but still positive. The argument of $ f_\lambda(z) $ is therefore between zero and $ \theta<\pi. $  We have shown that the analytic continuation of $ f_\lambda $ to the complex upper half plane has positive imaginary part, thus $ f_\lambda $ is operator monotone. 

The perspective function
\[
(t,s)\to s f_\lambda(ts^{-1})=\bigl(\lambda t^p+(1-\lambda) s^p\bigr)^{1/p}\qquad t,s>0
\]
is operator concave and so is any function that appears as the composition of an operator monotone function of one variable with the perspective. It follows that
\[
(t,s)\to\bigl(\lambda t^p +(1-\lambda)s^p\bigr)^{(1-p)/p}
\]
is operator concave. However, by an elementary calculation we may write
\[
\frac{t-s}{t^p-s^p}=\frac{1}{p}\int_0^1 \bigl(\lambda t^p+(1-\lambda)s^p\bigr)^{(1-p)/p}\,d\lambda,
\]
and the statement of the theorem follows.
\end{proof}

Take $ 0\le p\le 1. $ Since the function $ \displaystyle (t,s)\to\frac{t-s}{t^p-s^p} $ is operator concave,  it follows that the trace function
\[
(A,B)\to\tr K^*\frac{L_A-L_B}{L_A^p-R_B^p}(K)
\]
is concave in positive definite $ n\times n $ matrices for any $ n\times n $ matrix $ K, $ where $ L_A $ and $ R_B $ denote left and right multiplication with $ A $ and $ B. $\\[1ex]
By choosing $ K $ as the unit matrix we obtain:
\begin{theorem}
Let $ 0<p\le 1. $ The trace function
\[
(A,B)\to\tr\frac{A-B}{A^p-B^p}
\]
is concave in positive definite matrices.
\end{theorem}

\section{The Fréchet differential}

\begin{theorem}\label{Main concavity result for traces}
Consider the function $ f(t)=t^p $ for $ 0< p\le 1. $ The map
\[
x\to \tr h\,df(x)^{-1}h,
\]
defined for positive definite $ n\times n $ matrices $ x, $ is concave for each self-adjoint $ n\times n $ matrix $ h. $
\end{theorem}

\begin{proof}
Consider $ x>0 $ and a basis $ (e_i)_{i=1}^n $ in which $ x $ is diagonal with eigenvalues given by $ xe_i=\lambda_i e_i $ for $ i=1,\dots,n. $ We may then calculate
\[
e_i df(x) h e_j=e_i h e_j \frac{\lambda_i^p-\lambda_j^p}{\lambda_i-\lambda_j}\qquad i,j=1,\dots,n.
\]
Expressed in this basis  $ df(x)h=h\circ L_f\bigl(\lambda_1,\dots,\lambda_n\bigr) $ is the Hadamard product (entry-wise) product of $ h $ and the Löwner matrix 
\[
L_f\bigl(\lambda_1,\dots,\lambda_n\bigr)=  \left(\frac{\lambda^p_i-\lambda^p_j}{\lambda_i-\lambda_j}\right)_{i=1}^n .
\]
The inverse Fréchet differential $ df(x)^{-1}h $  is therefore well-defined and given by the Hadamard product
\[
df(x)^{-1} h =h\circ\left(\frac{\lambda_i-\lambda_j}{\lambda^p_i-\lambda^p_j}\right)_{i,j=1}^n
\]
expressed in the same basis and thus
\[
\tr h\,df(x)^{-1}h=\sum_{i,j=1}^n |(he_i\mid e_j)|^2\,\frac{\lambda_i-\lambda_j}{\lambda^p_i-\lambda^p_j}
=\tr h g(L_x, R_x)h,
\]
where $ L_x $ and $ R_x $ are left and right multiplication with $ x $ and
\[
g(t,s)=\frac{t-s}{t^p-s^p}\qquad t,s>0.
\]
The operators $ L_x $ and $ R_x $ are positive definite commuting operators on the Hilbert space $ \mathcal H=M_m $ equipped with the inner product $ (A,B)=\tr B^* A. $ The last expression $ \tr h\,df(x)^{-1}h=\tr h g(L_x, R_x)h $ is independent of any particular basis and since $ g $ is operator concave by Theorem~\ref{main operator concave function}, we obtain that
the map $ x\to\tr  h\,df(x)^{-1}h $ is concave by~\cite[Theorem 1.1]{kn:hansen:2006:3}.
\end{proof}

\begin{theorem}\label{Theorem: convexity of Q-form}
Consider the function $ f(t)=t^p $ for $ 0<p\le 1. $ The map
\[
(x,h)\to \tr h\, d f(x)h\qquad x>0, h^*=h
\]
of two variables is convex.
\end{theorem}

\begin{proof}
Keeping the notation as in the proof of Theorem~\ref{main theorem} we define two quadratic forms $ \alpha $ and $ \beta $ on $ \mathcal H\oplus \mathcal H $ by setting
\[
\begin{array}{rl}
\alpha(X\oplus Y)&=\lambda \tr X\, df(A_1)X +(1-\lambda) \tr Y\, df(A_2)Y\\[2ex]
\beta(X\oplus Y)&=\tr (\lambda X+(1-\lambda)Y)\, d f(A)(\lambda X+(1-\lambda)Y),
\end{array}
\]
where $ A_1,A_2 $ are positive definite matrices, and $ A=\lambda A_1+(1-\lambda) A_2 $ for some $ \lambda\in[0,1]. $
The statement of the theorem is equivalent to the majorisation
\begin{equation}\label{quadratic form majorisation}
\beta(X\oplus Y)\le\alpha(X\oplus Y)
\end{equation}
for arbitrary self-adjoint $ X,Y\in M_n\,. $ The quadratic form $ h\to\tr h\, d f(x)h $ is positive definite since     
\[
\tr h\, d f(x)h=\sum_{i,j=1}^n |(h e_i\mid e_j)|^2\frac{\lambda^p_i-\lambda^p_j}{\lambda_i-\lambda_j}\,,
\]           
where $ (e_i)_{i=1}^n $ is a basis in which $ x $ is diagonal and $ \lambda_1,\dots\lambda_n $ are the corresponding eigenvalues counted with multiplicity. We also notice that the corresponding sesqui-linear form is given by
\[
(h,h')\to\tr h'\, d f(x)h.
\]
The two quadratic forms $ \alpha $ and $ \beta $ are in particular positive definite. Therefore, there exists an operator $ \Gamma $ on $ \mathcal H\oplus\mathcal H $ which is positive definite in the Hilbert space structure given by $ \beta $ such that
\[
\alpha\bigl(X\oplus Y, X'\oplus Y'\bigr)=\beta\bigl(\Gamma(X\oplus Y), X'\oplus Y'\bigr)\qquad X,X',Y,Y'\in M_n\,,
\]
where we retain the notation $ \alpha $ and $ \beta $ also for the corresponding sesqui-linear forms.
Suppose $ \gamma $ is an eigenvalue of $ \Gamma $ corresponding to an eigenvector $ X\oplus Y. $ Then
\[
\alpha\bigl(X\oplus Y,X'\oplus Y'\bigr)=\beta\bigl(\gamma(X\oplus Y), X'\oplus Y'\bigr)\qquad\text{for}\quad  X',Y'\in M_n 
\]
or equivalently
\[
\begin{array}{l}
\lambda\tr X'\,d f(A_1)X+(1-\lambda)\tr Y'\,d f(A_2)Y\\[1.5ex]
=\gamma\tr (\lambda X'+(1-\lambda)Y')\, d f(A)(\lambda X+((1-\lambda)Y)
\end{array}
\]
for arbitrary $ X',Y'\in M_n\,. $ From this we may derive the identities
\[
d f(A_1)X=\gamma\, d f(A)(\lambda X+(1-\lambda)Y)=d f(A_2)Y
\]
and thus by setting $ M=d f(A)(\lambda X+(1-\lambda)Y), $ we obtain
\[
\begin{array}{rl}
d f(A)^{-1}(M)&=\lambda X+(1-\lambda) Y\\[1.5ex]
&=\lambda\, d f(A_1)^{-1}(\gamma M)+(1-\lambda)\, d f(A_2)^{-1}(\gamma M).
\end{array}
\]
By multiplying from the left with $ M^* $ and taking the trace we obtain
\[
\begin{array}{l}
\gamma\bigl(\lambda\tr M^* d f(A_1)^{-1} M + (1-\lambda)\,\tr M^* d f(A_2)^{-1}M\bigr)
=\displaystyle\tr M^* df(A)^{-1} M\\[2.5ex]
\ge\displaystyle \lambda\tr M^* df(A_1)^{-1} M +(1-\lambda)\tr M^* df(A_2)^{-1} M, 
\end{array}
\]
where the last inequality is implied by the concavity result in Theorem~\ref{Main concavity result for traces}.
This shows that the positive definite operator $ \Gamma\ge 1$ from which (\ref{quadratic form majorisation}) and the statement of the theorem follows.
\end{proof}

Since the dependence of the function $ f $ in $ \tr h\,df(x)h $ is linear we immediately obtain

\begin{corollary}

Let $ f $ be a function written on the form
\[
f(t)=\int_0^1 t^p\, d\mu(p)\qquad t>0,
\]
where $ \mu $ is a positive measure on the unit interval. Then the map
\[
(x,h)\to \tr h\, d f(x)h\qquad x>0, h^*=h
\]
of two variables is convex.
\end{corollary}

If we in the corollary above choose $ \mu $ as the Lebesgue measure, we realise that
\[
f(t)=\frac{t-1}{\log t}\qquad t>0
\]
is an example of a function such that $ (x,h)\to \tr h\, d f(x)h $ is convex.
Moreover, the perspective $ g $ of $ f $ given by
\[
g(t,s)=s f(ts^{-1})=s\frac{ts^{-1}-1}{\log(ts^{-1})}=\frac{t-s}{\log t-\log s}\qquad t,s>0
\]
is operator concave, and since
\[
\tr h\, d\log(x)^{-1} h=\tr h g(L_x,R_x)h
\]
this directly shows that the function $ x\to\tr h\,d\log(x)^{-1}h $ is concave, cf. \cite[equation (3.4)]{kn:lieb:1973:1}.

\section{More trace functions}

\begin{lemma}
Let $ K $ be a contraction. Then
\[
\psi(A)=q(A^{q-1}-K(K^*AK)^{q-1}K^*)\ge 0
\]
for $ -1\le q\le 1. $
\end{lemma} 

\begin{proof} By continuity we may assume $ K $ invertible.
For $ 0\le q\le 1 $ we use the inequality
\[
K^* A^{1-q} K\le (K^*AK)^{(1-q)},
\]
or by inversion
\[
K^{-1}A^{-(1-q)}(K^*)^{-1}=(K^* A^{1-q} K)^{-1}\ge (K^*AK)^{-(1-q)}
\]
which implies that
\[
A^{q-1}-K(K^*AK)^{q-1}K^*\ge 0.
\]
For $ -1\le q\le 0 $ we apply Jensen's sub-homogeneous operator inequality
\[
K^* A^{1-q} K\ge (K^*AK)^{1-q},
\]
or by inversion
\[
K^{-1}A^{-(1-q)} (K^*)^{-1}=(K^* A^{1-q} K)^{-1}\le (K^*AK)^{q-1}
\]
which implies that
\[
A^{q-1}\le K(K^*AK)^{q-1}K^*
\]
\end{proof}

\begin{corollary}
Let $ K $ be a contraction. The mapping
\[
 \varphi(A)=\tr (K^*AK)^q-\tr A^q\qquad A>0
\] 
is decreasing for $ -1\le q\le 1. $
\end{corollary}

\begin{proof}
The Fréchet differential of $ \varphi(A) $ is given by
\[
d\varphi(A)D=-q\tr\bigl(A^{q-1}- K(K^*AK)^{q-1}K^*\bigr)D=-\tr \psi(A) D,
\]
thus $ d\varphi(A)D\le 0 $ for arbitrary $ D\ge 0 $ by the preceding lemma.
\end{proof}

\begin{acknowledgement}
We thank Peter Harremoës for pointing out that the convexity of the residual entropy of a compound system may be easily inferred by considering it as a sum of relative entropies.
\end{acknowledgement}

{


      \noindent Frank Hansen: Institute for International Education, Tohoku University, Japan. Email: frank.hansen@m.tohoku.ac.jp.
      }

\end{document}